\PassOptionsToPackage{x11names,table,dvipsnames}{xcolor}
\documentclass[acmtog,balance=false,authorversion=true]{acmart}

\usepackage[ruled]{algorithm2e}
\usepackage{graphicx}
\usepackage{cancel}
\usepackage{wrapfig}

\AtBeginDocument{%
  }

\copyrightyear{2025}
\acmYear{2025}
\setcopyright{acmlicensed}\acmConference[SIGGRAPH Conference Papers '25]{Special
Interest Group on Computer Graphics and Interactive Techniques Conference
Conference Papers }{August 10--14, 2025}{Vancouver, BC, Canada}
\acmBooktitle{Special Interest Group on Computer Graphics and Interactive
Techniques Conference Conference Papers (SIGGRAPH Conference Papers '25), August
10--14, 2025, Vancouver, BC, Canada}
\acmDOI{10.1145/3721238.3730725}
\acmISBN{979-8-4007-1540-2/2025/08}

\newtheorem*{theorem*}{Theorem}

\begin{document}

\title{Stochastic Barnes-Hut Approximation for Fast Summation on the GPU}

\author{Abhishek Madan}
\affiliation{%
  \institution{University of Toronto}
  \country{Canada}}
\email{amadan@cs.toronto.edu}

\author{Nicholas Sharp}
\affiliation{%
  \institution{NVIDIA}
  \country{USA}}
\email{nsharp@nvidia.com}

\author{Francis Williams}
\affiliation{%
  \institution{NVIDIA}
  \country{USA}}
\email{fwilliams@nvidia.com}

\author{Ken Museth}
\affiliation{%
  \institution{NVIDIA}
  \country{USA}}
\email{kmuseth@nvidia.com}

\author{David I.W. Levin}
\affiliation{%
  \institution{University of Toronto and NVIDIA}
  \country{Canada}}
\email{diwlevin@cs.toronto.edu}

\newcommand{\bp}{\mathbf{p}}
\newcommand{\bq}{\mathbf{q}}
\newcommand{\Src}{\mathcal{S}}
\newcommand{\Q}{\mathcal{Q}}
\newcommand{\Tree}{T}
\newcommand{\bbox}{B}
\newcommand{\mm}{\tilde{m}}
\newcommand{\com}{\tilde{\bp}}
\newcommand{\Children}{C}
\newcommand{\Contrib}{D}
\newcommand{\Path}{\mathcal{P}}
\newcommand{\FF}{\tilde{\beta}}
\newcommand{\prr}{p_\textnormal{rr}}
\newcommand{\pa}{p_\textnormal{agg}}

\newcommand{\francis}[1]{\textcolor{green}{[\textbf{Francis: #1}]}}
\newcommand{\abhishek}[1]{\textcolor{blue}{[\textbf{Abhishek: #1}]}}
\newcommand{\ken}[1]{\textcolor{orange}{[\textbf{Ken: #1}]}}

\newcommand{\dl}[1]{\textcolor{red}{[\textbf{Dave: #1}]}}

\renewcommand{\shortauthors}{Madan et al.}

\begin{abstract}
We present a novel stochastic version of the Barnes-Hut approximation.
Regarding the level-of-detail (LOD) family of approximations as control variates, we construct an unbiased estimator of the kernel sum being approximated.
Through several examples in graphics applications such as winding number computation and smooth distance evaluation, we demonstrate that our method is well-suited for GPU computation, capable of outperforming a GPU-optimized implementation of the deterministic Barnes-Hut approximation by achieving equal median error in up to 9.4x less time.
\end{abstract}


\begin{teaserfigure}
  \includegraphics[width=\textwidth]{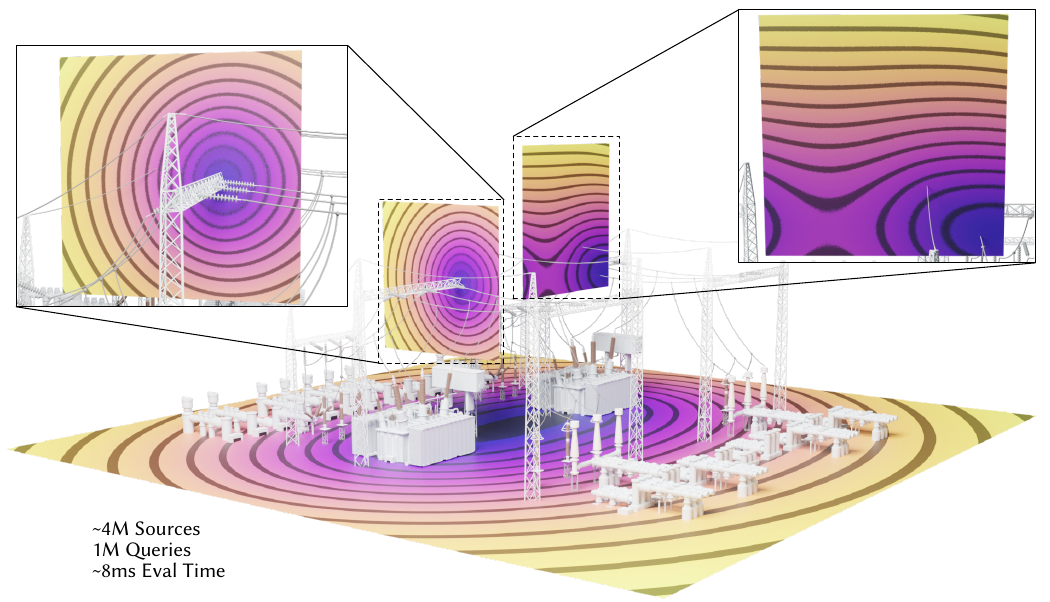}
  \caption{We compute the electrical potential induced by a power station on a GPU, by taking $2^{22}$ surface samples and evaluating the potential field on a $1000^2$ slice plane. Using one sample per tree subdomain, our stochastic method computes the effect of over 4 trillion particle-particle interactions in 8ms with almost no visible artifacts, while a GPU-optimized Barnes-Hut implementation takes 14ms and brute force takes 2800ms. \copyright{} Electrical substation model by CG\_loma under TurboSquid 3D Model License --- Standard.}
  \label{fig:teaser}
\end{teaserfigure}

\maketitle

\section{Introduction}
Fast summation lies at the heart of numerous practical algorithms in science, engineering, and computer graphics.
For example, $N$-Body problems rely on fast summation schemes to approximate gravitational forces~\cite{barnes1986hierarchical,greengard1987fast}, winding number computation for robust inside-outside testing can be accelerated using fast summation techniques~\cite{barill2018fast}, and solutions to linear partial differential equations (PDEs) via boundary element-type approaches require fast summation methods to evaluate the solution across the domain~\cite{martinsson2019fast}.

In computer graphics, stochastic estimation of integrals, which use random sampling to form probabilistic estimates instead of deterministic approximations, have become increasingly popular mathematical and practical tools.
Despite the noise artifacts produced by such methods at low sample counts, their ability to handle highly complex problems, and the tendency for the noise to flatten out with more samples, has been used to great effect in rendering and light transport simulation.
More recently, these approaches have also been applied effectively to problems in geometry processing and to the solution of linear PDEs.
Stochastic approaches have controllable running costs by user-defined sample budgets, and are often embarrassingly parallel, thus mapping well to modern GPU architectures and yielding algorithms which outperform classical approaches at large scales.

Unfortunately, stochastic methods for integration/summation can suffer from poor convergence rates in the number of samples.
Control variates (approximations of the integral) are a variance reduction technique that can circumvent slow convergence by giving low-sample regimes more acceptable error~\cite{mcbook}.
Introducing a control variate means that only the residual between the true integrand and the control variate needs to be estimated.
The better the approximation provided by the control variate, the smaller this residual, which then allows fewer samples in the estimate for comparable error.
However, it can be difficult to design a control variate by hand, limiting their use to problems with computationally efficient analytic approximations (although some neural approaches have recently been proposed to expand their applicability~\cite{muller2020neural,li2024neural}).

This work begins with the observation that in many of the new application domains for stochastic methods such as geometry processing, our integral/summation equations are much more structured, and many fast deterministic algorithms exist for their evaluation.
Our motivating insight is that these fast, deterministic, approximate methods can serve as robust and accurate control variates for stochastic algorithms.

We apply this observation to computing fast approximate summations.
Concretely, we propose an unbiased stochastic estimator for computing spatially varying fields which result from summation of a large number of globally supported kernel functions (e.g., gravitational potentials, winding numbers, smooth distances), using the traditional Barnes-Hut method as a control variate~\cite{barnes1986hierarchical}.
Standard Barnes-Hut is difficult to efficiently parallelize on the GPU due to the deep tree traversals required.
We show that using a Barnes-Hut traversal as a control variate to a stochastic estimator for the the full summation yields an algorithm that can provide up to an \textbf{9.4x} performance improvement for identical median errors when running on an NVIDIA RTX 4090 GPU.

The contributions of this work include a new use of a classical numerical method as a control variate, a proof that the resulting stochastic numerical method is unbiased, and the demonstration of our approach on a number of graphics-centric applications. 
An open-source implementation is provided on the project website.\footnote{https://www.dgp.toronto.edu/projects/stochastic-barnes-hut/}

\section{Related Work}

Fast summation algorithms are crucial in several fields of science, and as a result, many computational approximation methods have stemmed from these fields.
The classic example is $N$-body simulations: given $N$ particles that each exert some force on all other particles, how does a physical simulation evolve?
In computational physics, the Barnes-Hut approximation~\cite{barnes1986hierarchical} was developed to speed up $N$-body gravitational dynamics simulations, leading to $O(N \log N)$ computational complexity for the entire simulation or $O(\log N)$ to compute the force on a single particle.
Concurrently, the fast multipole method (FMM)~\cite{greengard1987fast} was developed to solve similar problems, but instead achieves $O(N)$ computational complexity, albeit at the cost of a significantly more complex preprocessing phase.
At their core, these methods rely on similar principles: they both use spatial data structures to organize the particles and multipole approximations to aggregate contributions from points with little impact on the sum.
However, the FMM takes it a step further and also computes polynomial approximations for each spatial region that approximate the total effect of distant multipole expansions, though higher-order approximations are required to obtain similar accuracy to Barnes-Hut~\cite{martinsson2019fast}.
The analytic work in deriving the multipole and local approximations makes it challenging to use the FMM as a general-purpose fast summation accelerator, and so we build upon Barnes-Hut, which can use first-order approximations which do not require such work.

\setlength{\columnsep}{0.7em}
\setlength{\intextsep}{0.01em}
\begin{wrapfigure}{r}{0.4\columnwidth}
  \centering
  \includegraphics[width=0.4\columnwidth]{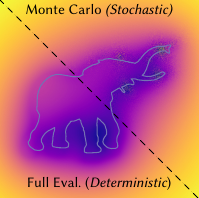}
\end{wrapfigure}
Neither the FMM nor Barnes-Hut are restricted to $N$-body problems --- they can both be applied to many problems where a set of source points interacts with a set of query points (e.g., samples from a polyline interacting with a grid of image pixels, see inset), with interaction magnitudes that decay with distance.
This generality has allowed both methods to find several applications in computer graphics as well, including generalized winding number acceleration~\cite{barill2018fast}, point set surfaces~\cite{alexa2001point}, fast linking number computation~\cite{qu2021fast}, repulsive curves and surfaces~\cite{yu2021rc,yu2021rs}, and smooth distances~\cite{madan2022smoothdists}.
Our algorithm would be able to accelerate any of these methods as an alternative fast summation technique, and is particularly suitable for enabling or improving GPU implementations of these methods.

Fast summations on the GPU are a classic algorithm that can be easily implemented in an efficient and query-parallel fashion by parallel loads into shared memory~\cite{nyland2007fast}.
Despite the simplicity of the brute-force GPU algorithm, the fast summation algorithms that work well on the CPU like Barnes-Hut are more challenging to port to the GPU.
The primary difficulty is in effectively parallelizing tree traversals, without introducing excessive thread execution divergence or being limited by memory access latency.
Before the advent of modern GPUs, many techniques were designed for general parallel or vector architectures: \citet{barnes1990modified} groups particles with similar traversals together to perform a single vectorized traversal per group; \citet{makino1990vectorization} proposes a stackless tree traversal and a breadth-first vectorized tree construction algorithm; and \citet{hernquist1990vectorization} parallelizes a single query by a breadth-first traversal.
More recently, \citet{burtscher2011efficient} proposed a complete pipeline for a Barnes-Hut implementation in CUDA using warp-voting traversals, and \citet{bedorf2012sparse} proposed a similar CUDA pipeline using breadth-first tree traversals.
Many similar techniques have also been developed for ray tracing on the GPU, particularly for query-parallel~\cite{wald2001interactive} and memory-efficient tree traversals~\cite{hapala2011efficient,vaidyanathan2019wide,museth21nanovdb}, and parallel tree builders~\cite{lauterbach2009fast,ylitie2017efficient}, culminating in software~\cite{parker2010optix} and hardware support for ray tracing on modern GPUs.
We take inspiration from many of these techniques in our work.

Monte Carlo methods are a simple yet effective means of estimating integrals using samples of the integrand~\cite{metropolis1949monte}, at the cost of introducing some noise (see inset above).
They have been used in physically-based rendering to estimate high-dimensional integrals that would otherwise be intractable, from the seminal work of \citet{kajiya1986rendering} to a plethora of techniques available today~\cite{pharr2023physically}.
More recently, Monte Carlo methods have also been used in other areas of graphics, such as solving volumetric PDEs~\cite{sawhney2020monte,sugimoto2023practical} and fluid simulation~\cite{rioux2022monte}.

Although most of these applications involve estimating continuous quantities, Monte Carlo methods can also estimate sums using samples of the summand.
This is particularly relevant for sampling the contributions of many lights in rendering, where there are a discrete number of light sources in the scene.
Many works broadly involve building a tree around the light sources in order to classify their importance~\cite{shirley1996monte}, and cluster lights together in a manner similar to Barnes-Hut~\cite{walter2005lightcuts,yuksel2019stochastic} or through an online learning process~\cite{wang2021learning}.
However, these approximations must deal with significant discontinuities due to occlusion, as well as anisotropy that prevents spatial clustering (instead using pre-existing lights as cluster representatives); neither of these are present in our problem setting, so we can perform more effective aggregation to reduce variance, in line with multipole approximations used in Barnes-Hut.
The Virtual Point Lights method~\cite{dachsbacher2014scalable,keller1997instant} does create lights that do not exist in the scene, but these exist along pre-computed light paths, and are primarily intended to unify complex lighting effects within a single framework, so they must be combined with the aforementioned many-light sampling methods to work well in a path tracer.

Departing from previous work in rendering and stochastic PDEs, we show that exploiting the structure of all-pairs summation problems and fusing fast, tree-based approximate methods with stochastic approaches yields a GPU-friendly algorithm that can outperform standard approaches in terms of error and performance in application-relevant regimes.

\section{Method}

\begin{figure}
    \centering
    \includegraphics[width=\columnwidth]{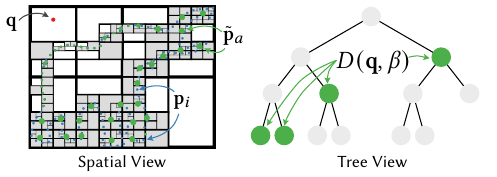}
    \caption{An illustration of a tree traversal performed by Barnes-Hut. Spatially (left), the centers of mass $\com_a$ used in the approximation are shown in green, while the source points $\bp_i \in \Src$ are shown in blue. On the tree data structure (right), the contribution nodes $D(\bq,\beta)$ are highlighted. The binary tree is for illustrative purposes only and does not correspond to the quadtree on the left.}
    \label{fig:bh}
\end{figure}
For the rest of the paper, we will operate in the following framework: given a kernel function $f(\bp, \bq)$, a set of source points $\Src = \{ \bp_1, \ldots, \bp_M \}$ with associated masses $\{ m_1, \ldots, m_M \}$, and a set of query points $\Q = \{ \bq_1, \ldots, \bq_N \}$, we want to evaluate
\begin{equation}
    \label{eq:kernel_sum}
    F(\bq) = \sum_{i = 1}^M m_j f(\bp_i, \bq)
\end{equation}
for all $\bq_j \in \Q$.
In typical applications, $f(\bp, \bq)$ is a function that depends on the distance between $\bp$ and $\bq$, $\| \bp - \bq \|$, typically decreasing in magnitude as the distance increases (though there are some exceptions, such as the Green's function of the 2D Laplacian and the Green's function of the 3D biharmonic operator).
Although this summation can be computed exactly for all $\bq_j$, it scales linearly with $M$ and $N$ and so becomes impractically slow as both values increase.
As a result, approximation methods are used in practice that take advantage of the distance-dependent behavior of $f$, such as the fast multipole method and the Barnes-Hut approximation.
Since our method is an extension of Barnes-Hut, we will briefly review it before describing our work.

\subsection{Overview of Barnes-Hut}
The Barnes-Hut approximation greatly reduces the cost of evaluating Eq.~\ref{eq:kernel_sum}, at the expense of some approximation error.
The key insight behind it is that at large distances from a query point $\bq$, many points in $\Src$ will have similar contributions to the total $F(\bq)$, which allows their contributions to be grouped together (Fig.~\ref{fig:bh}).
Many groups can be created, which greatly reduce the number of terms in the summation.

To facilitate the creation of such groups, the points in $\Src$ are placed in a subdivision data structure, such as an octree or a bounding volume hierarchy (Fig.~\ref{fig:bh}, left).
Typically, spatial subdivision structures are preferred for points (i.e., each tree node is characterized by the points it contains), so we will assume that we have an octree-like data structure, possibly with a wider per-dimension branching factor $d$ than the standard 2, and that leaf nodes contain only one source point.
We denote tree nodes by $\Tree_a$ and the set of children of $\Tree_a$ by $\Children(\Tree_a)$ (or equivalently, the set of child node indices); we also use $\Tree_a$ to denote the subset of points (and equivalently, point indices) contained in the node, so $\Tree_a = \bigcup_{\Tree_b \in \Children(\Tree_a)} \Tree_b$.
Each tree node has an associated axis-aligned bounding box containing all its points $\bbox_a$, with diameter $| \bbox_a |$, and also has an aggregate mass $\mm_a = \sum_{b \in \Tree_a} m_b$ and a center of mass $\com_a = \frac{\sum_{b \in \Tree_a} \mm_b \bp_b}{\mm_a}$.
To approximately evaluate $F(\bq)$, we traverse the tree starting from the root $\Tree_0$, and stop if either a \textit{far field condition} $\frac{\| \bq - \tilde{\bp}_a \|}{|\bbox_a|} \ge \beta$ is met (given a user-defined constant $\beta$) or if $\Tree_a$ is a leaf node, and add $\mm_a f(\com_a, \bq)$ to the total; otherwise we recursively continue the traversal at each of the node's children.
(Higher-order approximations can be used instead~\cite{martinsson2019fast}, but for ease of presentation we will not discuss them in detail.)
Intuitively, the far field condition is a scale-independent mechanism to check whether a cluster of particles is indistinguishable from a single, larger particle.
This tree traversal produces a set of tree nodes $\Contrib(\bq, \beta)$ that cover all the points in $\Src$, called the \textit{contribution nodes} of the query, and the approximated summation is then
\begin{equation}
    \label{eq:bh_sum}
    F_{BH}(\bq, \beta) = \sum_{\Tree_a \in \Contrib(\bq, \beta)} \mm_a f(\com_a, \bq).
\end{equation}
Pseudocode for Barnes-Hut is provided in the supplemental.

\subsection{Path Interpretation of Barnes-Hut}
When $\beta = \infty$, the Barnes-Hut algorithm traverses the entire tree and visits every node to obtain the exact answer from Eq.~\ref{eq:kernel_sum}.
Effectively, for each source point $\bp_i$, we traverse a path $\Path_i = [ \Tree_{i,0}, \ldots, \Tree_{i,d_i} ]$ down the tree, where $\Tree_{i,0} = \Tree_0$ and $d_i$ is the depth of the leaf node containing $\bp_i$.
Each tree node along a path represents an approximation of the true value of $m_i f(\bp_i, \bq)$, and traversing to the next node in the path is equivalent to replacing this approximation with a more accurate one.
The path traversal process can therefore be represented by a telescoping sum:
\begin{equation}
    \label{eq:single_path}
    F_{\Path_i}(\bq) = m_i f(\com_{i,0}, \bq) + \sum_{k = 1}^{d_i} m_i (f(\com_{i,k}, \bq) - f(\com_{i,k-1}, \bq)),
\end{equation}
and summing over all $M$ paths (one path per source point) gives $F(\bq) = \sum_{i = 1}^M F_{\Path_i}(\bq)$.
With smaller values of $\beta$, only part of the tree is traversed, and as a result the paths to the leaves are truncated.
Denoting a path prefix by $\Path_{i,k} = [ \Tree_{i,0}, \ldots, \Tree_{i,k} ]$ ($k \le d_i$), and the Barnes-Hut-induced truncation length by $\ell_i(\beta)$, Barnes-Hut traverses path prefixes $\Path_{i,\ell_i(\beta)}$, and we then have
\begin{equation}
    \label{eq:bh_path}
    F_{BH}(\bq, \beta) = \sum_i F_{\Path_{i,\ell_i(\beta)}}(\bq).
\end{equation}
Many paths share common prefixes, and will thus be identical after truncation, after which they can be aggregated into a single path with the combined mass of all the constituent paths.
Expanding the summands in Eq.~\ref{eq:bh_path} and grouping common prefixes gives Eq.~\ref{eq:bh_sum}, and thus path truncation and aggregation is equivalent to the typical truncated tree traversal perspective of Barnes-Hut.

\subsection{Path-Based Monte Carlo Estimator}
At first glance, there is no computational benefit to the path interpretation of Barnes-Hut --- computing telescoping sums is strictly more work than simply canceling terms and evaluating $f$ at the contribution nodes.
However, it opens up \textit{stochastic} methods to estimating Eq.~\ref{eq:kernel_sum}.
The telescoping sum in Eq.~\ref{eq:single_path}, for example, resembles an estimator that uses an infinite series of biased estimators to form an unbiased estimator~\cite{rhee2015unbiased,misso2022unbiased}.
Numerical approximations are in a sense biased estimators of the exact quantity --- although Barnes-Hut approximations from Eq.~\ref{eq:bh_sum} are deterministic, we can view them in a more general probabilistic sense, where it is clear that by construction their expected value does not match Eq.~\ref{eq:kernel_sum}.
Another way of viewing this is in terms of control variates, where the terms of the telescoping sum represent replacing a less accurate control variate for the estimator with a more accurate one.
Monte Carlo methods are most commonly deployed in settings where deterministic approximations are challenging to construct, such as in physically-based rendering, but we flip this perspective and instead use deterministic approximations as control variates to bootstrap a Monte Carlo estimator.
Not only does this produce an unbiased estimate of the exact quantity we wish to evaluate, but it also provides computational benefits on the GPU, as we will later demonstrate.

There are two primary random variables that we need to sample: the path index $I$ (equivalent to a source point index), and the path length $K$ (which depends on $I$).
For simplicity, we uniformly sample the path index among source point indices, and leave importance sampling schemes (e.g., based on distance to $\bq$) as future work.
The path length is sampled using a Russian roulette process, where at a node at depth $k$ along the path $\Path_i$, we decide to go to the next node with a probability $p_{i,k}(\bq)$.
In general, any non-zero probability will allow for an unbiased estimate, but in practice they must be carefully selected to keep the variance low.

\begin{figure}
    \centering
    \includegraphics[width=\columnwidth]{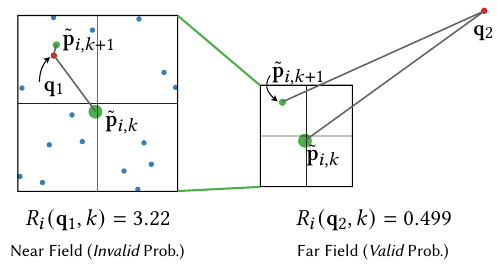}
    \caption{In the near field, $\com_{i,k+1}$ is much closer to $\bq_1$ than $\com_{i,k}$, so $R_i(\bq_1, k) > 1$, making it an invalid probability (left). However, in the far field, the distance from $\bq_2$ to $\com_{i,k}$ is similar to its distance to $\com_{i,k+1}$, so $R_i(\bq_2, k) \approx \frac{1}{d}$ (i.e., the reciprocal of the per-dimension branching factor), which is useful convergent behavior for Russian roulette probabilities (right).}
    \label{fig:ff_ratio}
\end{figure}

To sample path lengths, we take further inspiration from Barnes-Hut and use the \textit{far field ratio} $\FF_i(\bq, k) = \frac{\| \bq - \com_{i,k} \|}{|\bbox_{i,k}|}$ of a node $\Tree_{i,k}$ along $\Path_i$ as a key building block of the Russian roulette probabilities.
In Barnes-Hut, this ratio is compared with $\beta$ in the far field condition to determine whether or not a tree traversal continues to the node's children.
In the context of our algorithm, however, want to convert this deterministic decision into a stochastic decision that produces similar behavior on average (i.e., short paths in the far field, long paths in the near field).
To achieve this, we use a ratio of far field ratios $R_i(\bq, k) = \frac{\FF_i(\bq, k)}{\FF_i(\bq, k+1)}$ between the current and next node in the path to determine the Russian roulette probability.

To understand the significance of $R_i$, consider a query point $\bq$ that is far away from the source point $\bp_i$.
In an octree-like structure, the diameter of a child node $|\bbox_{i,k+1}|$ is $d$ times less than the diameter of its parent $|\bbox_{i,k}|$, where $d$ is the per-dimension branching factor.
Also, as we move further down the path $\Path_i$, the distance from $\bq$ to the parent's center of mass is similar to the distance from $\bq$ to the child's center of mass (i.e., $\| \bq - \com_{i,k} \|$ and $\| \bq - \com_{i,k+1} \|$ move closer together).
Therefore, for long paths through an octree-like structure, $R_i(\bq, k) \to \frac{1}{d}$ as $k$ grows (Fig.~\ref{fig:ff_ratio}, right).
Thus, $R_i(\bq, k)$ is a reasonable Russian roulette probability that disincentivizes deep paths in the far field without completely prohibiting them.

The near field (i.e., source points near $\bq$) requires a bit more care, since $\FF_i(\bq, k)$ may be quite small even at deep levels, and $R_i(\bq, k)$ may be greater than 1 at shallow levels if $\bq \in \bbox_{i,k}$, making it an invalid probability (Fig.~\ref{fig:ff_ratio}, left).
To solve these problems, we clamp the parent's far field ratio to be at least 1, to ensure that $\bq$ is outside of $\bbox_{i,k}$ and $\FF_i(\bq, k)$ exhibits far field behavior before we use it in $R_i$, and we clamp the ratio to be at most 1 to ensure it is a valid probability.
(We could use $R_i$ to simultaneously define a Russian roulette probability and a splitting factor when $R_i > 1$~\cite{Szirmay-KalosAS05}, but due to the difficulty of efficiently implementing path splitting on the GPU, we opted to exclusively use Russian roulette.)
The final Russian roulette probability is then
\begin{equation}
    \label{eq:rr_prob}
    p_{i,k}(\bq) = \min \left( 1, \frac{\max(1, \FF_i(\bq, k))}{\FF_i(\bq, k+1)} \right).
\end{equation}
In the supplemental material, we compare our Russian roulette probabilities with fixed probabilities of $p_{i,k} = \frac{1}{2}$ and completely disabling Russian roulette (i.e., $p_{i,k} = 1$).

After deciding to traverse to the next node in the path, we must include the corresponding term of the telescoping series in our running estimate.
Directly using Eq.~\ref{eq:single_path} is inefficient because it only accounts for the contribution of a single $\bp_i$ to Eq.~\ref{eq:kernel_sum}, and has high variance because a node's center of mass can vary quite significantly from any individual source point in the node.
The tree structure allows a solution to both of these problems: it makes it possible to aggregate several identical telescoping terms from different $F_{\Path_i}$ together in a single estimator term, and enables \textit{contribution swaps} between a parent node and the sum of all of its children:
\begin{equation}
    \label{eq:contrib_swap}
    \Delta_{i,k} = \left( \sum_{c \in \Children(\Tree_{i,k})} \mm_c f(\com_c, \bq) \right) - \mm_{i,k} f(\com_{i,k}, \bq).
\end{equation}
Essentially, this swap is a form of antithetic sampling, because the contribution of the child node along the current path is counterbalanced by the contributions of its sibling nodes.
Some example contribution swaps are illustrated in Fig.~\ref{fig:path_sample}.

\begin{figure}
    \centering
    \includegraphics[width=\columnwidth]{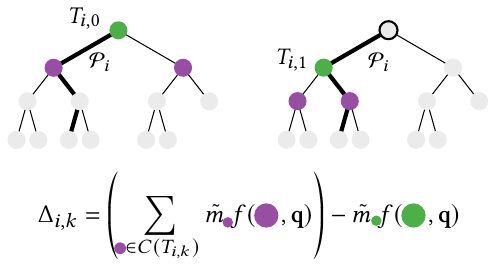}
    \caption{The first two contribution swaps along a path $\Path_i$ (bold). The parent node (green) subtracts its contribution from the sum of its child contributions (purple).}
    \label{fig:path_sample}
\end{figure}

The final estimator for a single truncated path sample $(I,K)$ is
\begin{equation}
    \label{eq:estimator}
    \hat{F}_1(\bq) = \mm_{I,0}f(\com_{I,0}, \bq) + \sum_{k=1}^K \frac{\Delta_{I,k-1}}{p(I \in \Tree_{I,k-1})p(K \ge k)},
\end{equation}
where $p(I \in \Tree_{i,k}) = \sum_{j \in \Tree_{i,k}} p(I = j)$, and $p(K \ge k) = \sum_{\ell=k}^{d_i} p(K = \ell) = \prod_{j = 0}^{k-1} p_{i,j}(\bq)$.
The $S$-sample estimator is $\hat{F}_S(\bq) = \frac{1}{S} \sum_{i=1}^S \hat{F}_1(\bq)$.

\begin{theorem}
\label{thm:unbiased}
Eq.~\ref{eq:estimator} is an unbiased estimator of Eq.~\ref{eq:kernel_sum}; that is, $E \left[ \hat{F}_1(\bq) \right] = F(\bq)$.
\end{theorem}
The proof of Theorem~\ref{thm:unbiased} is in the supplemental material.

\subsubsection{Domain Stratification}
Stratified sampling is an important variance reduction technique in Monte Carlo estimators that reduces the effect of samples clumping in similar regions.
Analogously, stratified sampling of different paths helps us achieve good coverage of the terms in Eq.~\ref{eq:kernel_sum}.
Spatial data structures already provide a convenient means of domain stratification with the different levels of the tree that represent nested subdomains (e.g., Fig.~\ref{fig:bh}, left), so we can easily achieve stratified samples by modifying the path roots.
In other words, rather than starting all paths from the tree root $\Tree_0$, we can instead define a set of descendant nodes, or subdomains, as path roots (as long as they cover all of $\Src$), and sample paths starting from these nodes. 
For simplicity, we simply use the direct children of the root $\Tree_0$ as the subdomains, but other choices are also possible, such as the Barnes-Hut contribution nodes $\Contrib(\bq, \beta)$ for a low $\beta$, though we leave the GPU-efficient implementation of such schemes as future work.

Algorithm pseudocode is provided in the supplemental.

\section{Results}

\begin{figure*}
\includegraphics[width=\textwidth]{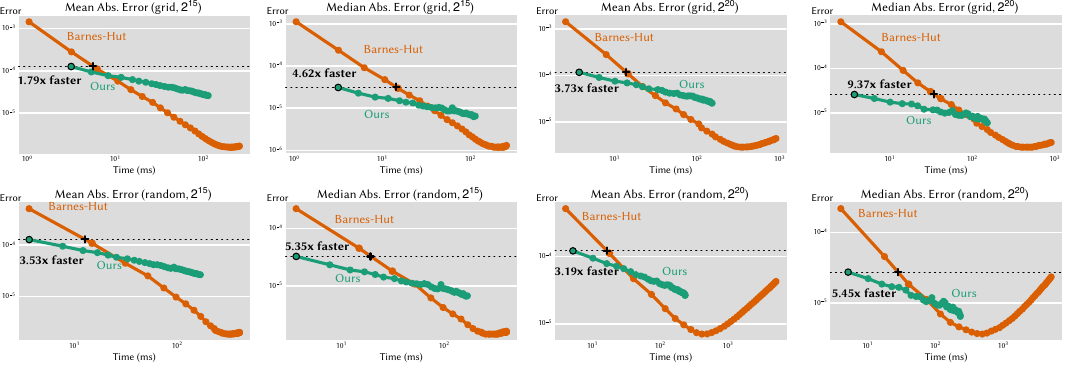}
\caption{
Convergence of our method vs Barnes-Hut, for 1 million random and grid query points, and $2^{15}$ and $2^{20}$ source points sampled from the Stanford bunny. Each point on the curves for our method (green) correspond to the number of samples per subdomain $S$ from 1 to 32; each point on the curves for Barnes-Hut (orange) correspond to the accuracy parameter $\beta$ from 1 to 40. The $S=1$ point on each curve for our method is outlined in black, and the point on the Barnes-Hut curve that achieves the same error is also marked (found by linear interpolation on the log-log plot). Across all source point sets and query distributions, and for both mean and median error, our method is roughly 2-9 times faster than Barnes-Hut. The Barnes-Hut curves turn upwards for large $\beta$ due to floating-point error for computing sums with many terms in single-precision.
}
\label{fig:convergence}
\end{figure*}

\subsection{Implementation Details}
We implemented our method and Barnes-Hut using C++ and CUDA, using the Eigen linear algebra library~\cite{eigen}.
To our knowledge, there is no open-source, general-purpose GPU implementation of Barnes-Hut, so we opted to implement and optimize our own baseline.
Both methods use a common octree-like spatial data structure, that supports variable branching factors.
The tree builder is currently CPU-only, and we leave incorporating fast GPU tree builders (e.g., the IndexGrid builder from fVDB~\cite{williams2024fVDB}) as future work.
We implemented Barnes-Hut with a stackless BVH traversal~\cite{makino1990vectorization,hapala2011efficient} and warp voting~\cite{barnes1990modified,burtscher2011efficient} to simultaneously reduce thread divergence and improve accuracy, using a tree with a per-dimension branching factor of $d=2$ (i.e., a total branching factor of 8).
The implementation of our method maintains two separate RNG states for drawing path index samples and Russian roulette samples independently, and seeds them both identically across warps to reduce thread divergence.
Our method also uses a tree with a wider per-dimension branching factor of $d=4$ (i.e., a total branching factor of 64) compared to Barnes-Hut, as we found that a narrow branching factor for Barnes-Hut allowed the algorithm to make efficiently make fine-grained far field decisions, while our algorithm performs better with wider trees that balance path depth with the cost of scanning over all of a node's children for contribution swaps.
To maximize performance, all operations are implemented using single-precision floating point.
Performance testing was done using a Dell Precision Workstation equipped with an AMD Threadripper CPU (Model 7955WX), 64 GB of RAM, and an NVIDIA RTX 4090 GPU with 24 GB of memory.

All geometric data used in the rest of this section is isotropically rescaled to fit in a $[-1,1]^3$ grid.

\subsection{Comparison to Barnes-Hut}
To start, we used a Coulomb-like potential $f(\bp, \bq) = -\frac{1}{\| \bp - \bq \|}$ to investigate the performance of our algorithm compared to Barnes-Hut.
We took $2^{15}$ (32,768) and $2^{20}$ (1,048,576) samples on the Stanford bunny, and evaluated both our method and Barnes-Hut on a grid of $100^3$ points in $[-1,1]^3$, as well as $1,000,000$ randomly distributed points in the same cube for a variety of parameter values.
In Barnes-Hut, the main parameter is the far field threshold $\beta$, and our method provides the number of samples per subdomain (spsd), denoted $S$.
In Fig.~\ref{fig:convergence}, we plotted the mean and median absolute error versus execution time for both source sizes and both query distributions, using parameter sweeps $\beta \in [1,2,\ldots,40]$ and $S \in [1,2,\ldots,32]$.

\setlength{\columnsep}{0.7em}
\setlength{\intextsep}{0.01em}
\begin{wrapfigure}{r}{0.5\columnwidth}
  \centering
  \includegraphics[width=0.5\columnwidth]{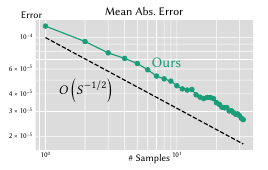}
\end{wrapfigure}
Across all configurations, Barnes-Hut can only match the mean and median error achieved by our method with $S=1$ after 2-8 times more execution time.
Furthermore, our relative speedup is faster on random query sets, where Barnes-Hut is slowed down by significant thread divergence (though warp voting allows it to use that time to drive down the overall error).
The relative lack of sensitivity of our method to the query set's spatial coherence is a remarkable property that is achieved through probabilistic path traversals, which change far less across query points than overall tree traversals.
Our method also achieves higher relative speedups under median error compared to mean error, but this suggests that much of the error in our method comes from isolated outliers rather than spread out across all queries.
In rendering, this is caused by ``fireflies'' due to importance sampling schemes giving high weight to some outlier samples that are far from the true expected value, so improving our sampling method and even using a specially designed denoiser for these noise patterns~\cite{vicini2019denoising} are potential ways to address this.
Another aspect of our method is that we converge at a slower than Barnes-Hut --- although we have a performance advantage at lower sample counts, this advantage disappears as $S$ increases because the error converges at a rate the Monte Carlo rate of $O(S^{-1/2})$ (though since we measure the mean and median absolute error rather than the root mean squared error, this rate is not precise in our case, see inset).
Our method is therefore most useful with low sample counts,
though regardless the low-sample regime is the main regime of interest, since we are estimating a quantity that can be evaluated in finite time.

We also tested our method on a variety of source distributions, using recommended parameter configurations for both Barnes-Hut and our method.
Papers in the Barnes-Hut literature in graphics recommend using $\beta=2$~\cite{barill2018fast,qu2021fast,madan2022smoothdists},
as this usually achieves a good tradeoff between efficient computation and acceptable error in graphics applications,
so we make the same choice.
As for our method, approximating a discrete summation puts a hard limit on the sample count, as there is a sample threshold where our method becomes more expensive than a full evaluation, and as seen in Fig.~\ref{fig:convergence}, our absolute error converges slowly; therefore, we pick $S=1$ as a reasonable default.
We evaluated our method with these choices across a mesh dataset from \citet{MPZ14}, where we drew $2^{15}$, $2^{17}$, and $2^{20}$ uniformly distributed samples from each surface to use as kernel source points (Table~\ref{tbl:dataset}).
Although our performance is comparable with Barnes-Hut for a small number of sources ($2^{15}$), our method scales much better, becoming 2x faster at the chosen settings with $2^{20}$ source samples, while maintaining lower mean and median error across all source sizes $M$.
Our maximum error is always an order of magnitude larger than Barnes-Hut, but the significantly lower median error suggests that these points are outliers.
\begin{table*}
    \centering
    \rowcolors{2}{white}{CornflowerBlue!25}
    \caption{Timings and errors of our method and Barnes-Hut on a dataset of 116 meshes, run with typical parameter settings. Barnes-Hut is evaluated with $\beta=2$, and our algorithm is evaluated with $S=1$ (i.e., one sample per subdomain). We vary the number of source samples $M$ to be $2^{15}$ (32,768), $2^{17}$ (131,072), and $2^{20}$ (1,048,576), draw them from each mesh surface in the dataset, and their total gravitational potential kernel on a grid of $100^3$ points. The mean and standard deviation for each error statistic across the dataset are reported, and the better metric is highlighted in bold at each source set size. At typical algorithm parameter settings, we match our outperform an optimized Barnes-Hut implementation by over 2x, while always achieving roughly 5x lower mean error and 17x lower median error.}
    \begin{tabular}{l r r r r}
        \toprule
        \textit{Algorithm (Typical Config.)} & \textit{Time (ms)} & \textit{Mean Abs. Err.} & \textit{Median Abs. Err} & \textit{Max Abs. Err.} \\
        Barnes-Hut $\left( \beta=2, M=2^{15} \right)$ & \textbf{3.48 $\pm$ 0.93} & 5.59e-04 $\pm$ 2.83e-04 & 4.50e-04 $\pm$ 2.37e-04 & \textbf{3.72e-03 $\pm$ 1.77e-03} \\
        Ours $\left( S=1, M=2^{15} \right)$ & 3.63 $\pm$ 0.97 & \textbf{9.91e-05 $\pm$ 2.28e-05} & \textbf{2.51e-05 $\pm$ 6.73e-06} & 5.61e-02 $\pm$ 3.05e-02 \\
        \midrule
        Barnes-Hut $\left( \beta=2, M=2^{17} \right)$ & 5.00 $\pm$ 1.46 & 5.60e-04 $\pm$ 2.82e-04 & 4.50e-04 $\pm$ 2.36e-04 & \textbf{3.73e-03 $\pm$ 1.75e-03} \\
        Ours $\left( S=1, M=2^{17} \right)$ & \textbf{3.92 $\pm$ 1.04} & \textbf{9.81e-05 $\pm$ 2.21e-05} & \textbf{2.43e-05 $\pm$ 6.22e-06} & 5.62e-02 $\pm$ 3.23e-02 \\
        \midrule
        Barnes-Hut $\left( \beta=2, M=2^{20} \right)$ & 8.61 $\pm$ 2.89 & 5.60e-04 $\pm$ 2.82e-04 & 4.50e-04 $\pm$ 2.36e-04 & \textbf{3.74e-03 $\pm$ 1.76e-03} \\
        Ours $\left( S=1, M=2^{20} \right)$ & \textbf{4.04 $\pm$ 1.12} & \textbf{9.77e-05 $\pm$ 2.16e-05} & \textbf{2.40e-05 $\pm$ 5.66e-06} & 5.60e-02 $\pm$ 3.13e-02 \\
        \bottomrule
    \end{tabular}
    \label{tbl:dataset}
\end{table*}

A useful theoretical property of our approach is that it is an unbiased estimator of the true summation value, as opposed to a deterministic but biased estimate of the quantity.
Fig.~\ref{fig:error} looks at the same slice as Fig.~\ref{fig:teaser} and examines the log error of both our method and Barnes-Hut.
Barnes-Hut error has ring-like artifacts that indicate where the far field condition caused the set of contribution nodes $D(\bq, \beta)$ to change in space; our method, meanwhile, does not have such artifacts, but error is more highly concentrated in regions with larger gradients.
Monte Carlo noise is also visible in our error.
\begin{figure}
\includegraphics[width=\columnwidth]{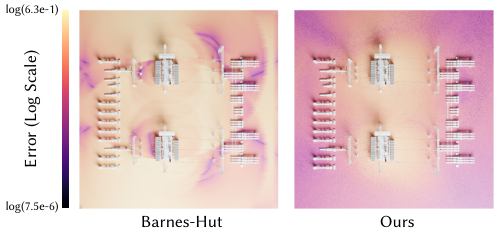}
\caption{In the example from Fig.~\ref{fig:teaser}, Barnes-Hut (left) exhibits discontinuous error patterns when the far field condition causes a change in contribution nodes, while our method (right) does not exhibit such artifacts and instead is more concentrated in regions where the field rapidly changes, while being lower on average than Barnes-Hut.}
\label{fig:error}
\end{figure}

\subsection{Applications}
We applied our method to a variety of different kernel functions, and compared the visual result quality from our method to the ground truth results computed with an exact GPU-optimized algorithm~\cite{nyland2007fast}.
In all examples, slice planes are evaluated as $1000^2$ grids for 1 million query points.
The identically-seeded RNG states across warps cause visually displeasing spatial coherence, so we shuffle the grid positions before passing them to our algorithm.
Although this induces a slight performance hit due to additional thread divergence, it is much smaller than a comparable hit to performance would be to Barnes-Hut (Fig.~\ref{fig:convergence}), since we retain most of the benefits of identical RNG states.

\subsubsection{Coulomb Potentials}
In Fig.~\ref{fig:teaser}, we show a visualization of an electric (Coulomb) potential field induced by $2^{22}$ (over 4 million) points on the surface of an electrical substation, with some points given a higher mass (i.e., charge) in the cables above the main units in the center.
With just one sample per subdomain, we produce a smooth result on all evaluation slices in roughly 8 milliseconds for the large slice plane and 6 milliseconds for the smaller planes.

\subsubsection{Winding Numbers}
In Fig.~\ref{fig:winding}, we compute winding numbers~\cite{jacobson2013winding} using our method from $2^{20}$ sources, with 1 sample per subdomain and 16 samples per subdomain, as well as the ground truth.
Our method mostly succeeds in classifying inside and outside cells with $S=1$, but contains noise near the boundary where the winding number changes rapidly; using 16 samples per subdomain mostly alleviates this issue at a significant performance cost.
\begin{figure}
\includegraphics[width=\columnwidth]{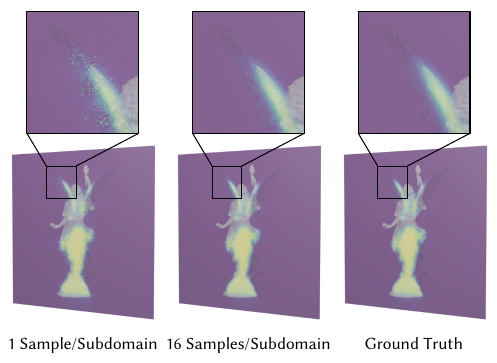}
\caption{Winding numbers computed via fast stochastic summation from $2^{20}$ sources. A full evaluation via brute force takes 887ms, while our method takes 7.80ms with 1 sample per subdomain, and 119ms with 16 samples per subdomain.}
\label{fig:winding}
\end{figure}

\subsubsection{Smooth Minimum Distances}
In Fig.~\ref{fig:smoothdist}, we compute smooth distance fields~\cite{madan2022smoothdists} from $2^{20}$ sources, using our method with 1 sample per subdomain and 64 samples per subdomain, and the ground truth.
More precisely, we estimate the sum of exponentials and apply a logarithm to the estimate as one would do for the full summation; however, unlike previous examples, the result is also biased for our method, because the expectation operator cannot be directly ``passed through'' the logarithm.
Unlike \citet{madan2022smoothdists}, we do not use off-center bounding boxes to guarantee the field is a conservative estimate of the true distance.
At $S=1$, our method has difficulty resolving the band near the zero level set, but can more effectively resolve it at $S=64$, though this is almost as expensive as a full evaluation.
\begin{figure}
\includegraphics[width=\columnwidth]{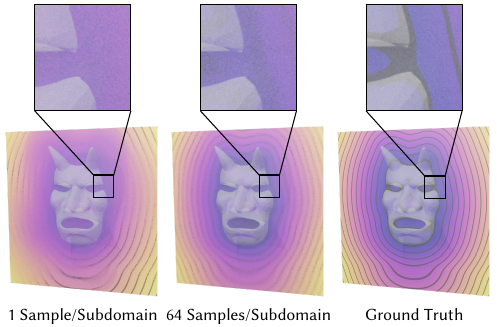}
\caption{Smooth distances computed via fast stochastic summation from $2^{20}$ source points. A full evaluation via brute force takes 710ms, and our method takes 7.10ms with 1 sample per subdomain 605ms with 64 samples per subdomain.}
\label{fig:smoothdist}
\end{figure}

\section{Conclusions and Future Work}
We presented a method that stochastically estimates kernel sums using a novel unbiased algorithm using a Barnes-Hut-esque hierarchy of control variates.
Our method already achieves excellent results on Coulomb potential evaluation problems, but there is significant room to improve the employed sampling techniques to achieve comparable results on other problems such as winding number and smooth distance evaluation.
We only use simple uniform sampling for paths, so there is an immediate opportunity to improve the sampled paths with importance sampling (though this could create thread divergence without care).
The Russian roulette probabilities also deserve further investigation using more principled approaches such as zero-variance theory~\cite{Vorba:2016:ARR} and estimator efficiency optimization~\cite{Rath2022} --- since there are a finite number of paths, it may be possible to use the aforementioned techniques to compute all probabilities for a given query point.
A unified Russian roulette and splitting framework on the GPU is another related direction of future work.
Recently, the ReSTIR algorithm~\cite{Wyman:2023:Gentle} has revolutionized real-time path tracing by reusing samples across pixels in space and time; integrating such an approach into our framework is another compelling direction.
Beyond discrete summations, applying this method to continuous integrals could be fruitful, potentially leading to an alternative to boundary element methods in the same manner Walk on Spheres provides an alternative to finite elements~\cite{sawhney2020monte}. 

Our method uses more compute than traditional tree-based approximations like Barnes-Hut, but as a result, it reduces thread divergence in a way that is relatively insensitive to query order --- a property that, to our knowledge, is unique to our method, providing it with unique advantages on the GPU.
Furthermore, stackless tree traversals require more memory traffic than their recursive counterparts to access nodes, so being able to perform simpler path traversals means that we access fewer tree nodes overall.
However, part of our computation involves a memory scan over an uncompressed buffer of node-aggregated data such as centers of mass and aggregated masses, which still increases memory traffic in conjunction with compute
and can cause scaling issues at higher sample counts,
though it is possible to reduce this through aggressive quantization and compression, as long as leaf data remains uncompressed to keep the estimator unbiased.
Compressing the tree nodes would also help Barnes-Hut, and further pursuing optimizations of both methods would change the performance tradeoff between the two.
More radical changes such as completely different data structures could also enable the use of ``single-term'' estimators~\cite{Mcleish2011,misso2022unbiased} that require cheap random access to internal tree nodes (as opposed to following child pointers in typical tree data structures), and would completely eliminate thread divergence on the GPU.
Overall, following the trend of other recent Monte Carlo methods, we believe there is a lot of potential in stochastic methods to improve upon deterministic methods, or at the very least, provide a dialogue between the two classes of methods to offer different choices for users depending on their needs, making it an exciting time to explore the use of stochasticity in traditionally deterministic computational domains.

\begin{acks}
This project was funded in part by an NSERC Discovery Grant (RGPIN-2023-05120) and an Ontario Early Researchers Award.
The first author was funded by an NSERC Canada Graduate Scholarship --- Doctoral.
We thank the attendees of the GraphQUON workshop for useful early discussions, and Mark Harris and Lars Nyland for proofreading.
\end{acks}

\bibliographystyle{ACM-Reference-Format}
\bibliography{stochastic_barnes_hut.bib}


\begin{thebibliography}{49}


\ifx \showCODEN    \undefined \def \showCODEN     #1{\unskip}     \fi
\ifx \showISBNx    \undefined \def \showISBNx     #1{\unskip}     \fi
\ifx \showISBNxiii \undefined \def \showISBNxiii  #1{\unskip}     \fi
\ifx \showISSN     \undefined \def \showISSN      #1{\unskip}     \fi
\ifx \showLCCN     \undefined \def \showLCCN      #1{\unskip}     \fi
\ifx \shownote     \undefined \def \shownote      #1{#1}          \fi
\ifx \showarticletitle \undefined \def \showarticletitle #1{#1}   \fi
\ifx \showURL      \undefined \def \showURL       {\relax}        \fi
\providecommand\bibfield[2]{#2}
\providecommand\bibinfo[2]{#2}
\providecommand\natexlab[1]{#1}
\providecommand\showeprint[2][]{arXiv:#2}

\bibitem[Alexa et~al\mbox{.}(2001)]%
        {alexa2001point}
\bibfield{author}{\bibinfo{person}{Marc Alexa}, \bibinfo{person}{Johannes
  Behr}, \bibinfo{person}{Daniel Cohen-Or}, \bibinfo{person}{Shachar
  Fleishman}, \bibinfo{person}{David Levin}, {and} \bibinfo{person}{Claudio~T
  Silva}.} \bibinfo{year}{2001}\natexlab{}.
\newblock \showarticletitle{Point set surfaces}. In
  \bibinfo{booktitle}{\emph{Proceedings Visualization, 2001. VIS'01.}} IEEE,
  \bibinfo{pages}{21--29}.
\newblock


\bibitem[Barill et~al\mbox{.}(2018)]%
        {barill2018fast}
\bibfield{author}{\bibinfo{person}{Gavin Barill}, \bibinfo{person}{Neil~G
  Dickson}, \bibinfo{person}{Ryan Schmidt}, \bibinfo{person}{David~IW Levin},
  {and} \bibinfo{person}{Alec Jacobson}.} \bibinfo{year}{2018}\natexlab{}.
\newblock \showarticletitle{Fast winding numbers for soups and clouds}.
\newblock \bibinfo{journal}{\emph{ACM Transactions on Graphics (TOG)}}
  \bibinfo{volume}{37}, \bibinfo{number}{4} (\bibinfo{year}{2018}),
  \bibinfo{pages}{1--12}.
\newblock


\bibitem[Barnes and Hut(1986)]%
        {barnes1986hierarchical}
\bibfield{author}{\bibinfo{person}{Josh Barnes} {and} \bibinfo{person}{Piet
  Hut}.} \bibinfo{year}{1986}\natexlab{}.
\newblock \showarticletitle{A hierarchical O (N log N) force-calculation
  algorithm}.
\newblock \bibinfo{journal}{\emph{nature}} \bibinfo{volume}{324},
  \bibinfo{number}{6096} (\bibinfo{year}{1986}), \bibinfo{pages}{446--449}.
\newblock


\bibitem[Barnes(1990)]%
        {barnes1990modified}
\bibfield{author}{\bibinfo{person}{Joshua~E Barnes}.}
  \bibinfo{year}{1990}\natexlab{}.
\newblock \showarticletitle{A modified tree code: Don't laugh; it runs}.
\newblock \bibinfo{journal}{\emph{J. Comput. Phys.}} \bibinfo{volume}{87},
  \bibinfo{number}{1} (\bibinfo{year}{1990}), \bibinfo{pages}{161--170}.
\newblock


\bibitem[B{\'e}dorf et~al\mbox{.}(2012)]%
        {bedorf2012sparse}
\bibfield{author}{\bibinfo{person}{Jeroen B{\'e}dorf},
  \bibinfo{person}{Evghenii Gaburov}, {and} \bibinfo{person}{Simon~Portegies
  Zwart}.} \bibinfo{year}{2012}\natexlab{}.
\newblock \showarticletitle{A sparse octree gravitational N-body code that runs
  entirely on the GPU processor}.
\newblock \bibinfo{journal}{\emph{J. Comput. Phys.}} \bibinfo{volume}{231},
  \bibinfo{number}{7} (\bibinfo{year}{2012}), \bibinfo{pages}{2825--2839}.
\newblock


\bibitem[Burtscher and Pingali(2011)]%
        {burtscher2011efficient}
\bibfield{author}{\bibinfo{person}{Martin Burtscher} {and}
  \bibinfo{person}{Keshav Pingali}.} \bibinfo{year}{2011}\natexlab{}.
\newblock \showarticletitle{An efficient CUDA implementation of the tree-based
  barnes hut n-body algorithm}.
\newblock In \bibinfo{booktitle}{\emph{GPU computing Gems Emerald edition}}.
  \bibinfo{publisher}{Elsevier}, \bibinfo{pages}{75--92}.
\newblock


\bibitem[Dachsbacher et~al\mbox{.}(2014)]%
        {dachsbacher2014scalable}
\bibfield{author}{\bibinfo{person}{Carsten Dachsbacher},
  \bibinfo{person}{Jaroslav K{\v{r}}iv{\'a}nek}, \bibinfo{person}{Milo{\v{s}}
  Ha{\v{s}}an}, \bibinfo{person}{Adam Arbree}, \bibinfo{person}{Bruce Walter},
  {and} \bibinfo{person}{Jan Nov{\'a}k}.} \bibinfo{year}{2014}\natexlab{}.
\newblock \showarticletitle{Scalable realistic rendering with many-light
  methods}. In \bibinfo{booktitle}{\emph{Computer Graphics Forum}},
  Vol.~\bibinfo{volume}{33}. Wiley Online Library, \bibinfo{pages}{88--104}.
\newblock


\bibitem[Greengard and Rokhlin(1987)]%
        {greengard1987fast}
\bibfield{author}{\bibinfo{person}{Leslie Greengard} {and}
  \bibinfo{person}{Vladimir Rokhlin}.} \bibinfo{year}{1987}\natexlab{}.
\newblock \showarticletitle{A fast algorithm for particle simulations}.
\newblock \bibinfo{journal}{\emph{Journal of computational physics}}
  \bibinfo{volume}{73}, \bibinfo{number}{2} (\bibinfo{year}{1987}),
  \bibinfo{pages}{325--348}.
\newblock


\bibitem[Guennebaud et~al\mbox{.}(2010)]%
        {eigen}
\bibfield{author}{\bibinfo{person}{Ga\"{e}l Guennebaud},
  \bibinfo{person}{Beno\^{i}t Jacob}, {et~al\mbox{.}}}
  \bibinfo{year}{2010}\natexlab{}.
\newblock \bibinfo{title}{Eigen v3}.
\newblock \bibinfo{howpublished}{http://eigen.tuxfamily.org}.
\newblock


\bibitem[Hapala et~al\mbox{.}(2011)]%
        {hapala2011efficient}
\bibfield{author}{\bibinfo{person}{Michal Hapala},
  \bibinfo{person}{Tom{\'a}{\v{s}} Davidovi{\v{c}}}, \bibinfo{person}{Ingo
  Wald}, \bibinfo{person}{Vlastimil Havran}, {and} \bibinfo{person}{Philipp
  Slusallek}.} \bibinfo{year}{2011}\natexlab{}.
\newblock \showarticletitle{Efficient stack-less bvh traversal for ray
  tracing}. In \bibinfo{booktitle}{\emph{Proceedings of the 27th Spring
  Conference on Computer Graphics}}. \bibinfo{pages}{7--12}.
\newblock


\bibitem[Hernquist(1990)]%
        {hernquist1990vectorization}
\bibfield{author}{\bibinfo{person}{Lars Hernquist}.}
  \bibinfo{year}{1990}\natexlab{}.
\newblock \showarticletitle{Vectorization of tree traversals}.
\newblock \bibinfo{journal}{\emph{J. Comput. Phys.}} \bibinfo{volume}{87},
  \bibinfo{number}{1} (\bibinfo{year}{1990}), \bibinfo{pages}{137--147}.
\newblock


\bibitem[Jacobson et~al\mbox{.}(2013)]%
        {jacobson2013winding}
\bibfield{author}{\bibinfo{person}{Alec Jacobson}, \bibinfo{person}{Ladislav
  Kavan}, {and} \bibinfo{person}{Olga Sorkine}.}
  \bibinfo{year}{2013}\natexlab{}.
\newblock \showarticletitle{Robust Inside-Outside Segmentation using
  Generalized Winding Numbers}.
\newblock \bibinfo{journal}{\emph{ACM Trans. Graph.}} \bibinfo{volume}{32},
  \bibinfo{number}{4} (\bibinfo{year}{2013}).
\newblock


\bibitem[Kajiya(1986)]%
        {kajiya1986rendering}
\bibfield{author}{\bibinfo{person}{James~T Kajiya}.}
  \bibinfo{year}{1986}\natexlab{}.
\newblock \showarticletitle{The rendering equation}. In
  \bibinfo{booktitle}{\emph{Proceedings of the 13th annual conference on
  Computer graphics and interactive techniques}}. \bibinfo{pages}{143--150}.
\newblock


\bibitem[Keller(1997)]%
        {keller1997instant}
\bibfield{author}{\bibinfo{person}{Alexander Keller}.}
  \bibinfo{year}{1997}\natexlab{}.
\newblock \showarticletitle{Instant radiosity}. In
  \bibinfo{booktitle}{\emph{Proceedings of the 24th annual conference on
  Computer graphics and interactive techniques}}. \bibinfo{pages}{49--56}.
\newblock


\bibitem[Lauterbach et~al\mbox{.}(2009)]%
        {lauterbach2009fast}
\bibfield{author}{\bibinfo{person}{Christian Lauterbach},
  \bibinfo{person}{Michael Garland}, \bibinfo{person}{Shubhabrata Sengupta},
  \bibinfo{person}{David Luebke}, {and} \bibinfo{person}{Dinesh Manocha}.}
  \bibinfo{year}{2009}\natexlab{}.
\newblock \showarticletitle{Fast BVH construction on GPUs}. In
  \bibinfo{booktitle}{\emph{Computer Graphics Forum}},
  Vol.~\bibinfo{volume}{28}. Wiley Online Library, \bibinfo{pages}{375--384}.
\newblock


\bibitem[Li et~al\mbox{.}(2024)]%
        {li2024neural}
\bibfield{author}{\bibinfo{person}{Zilu Li}, \bibinfo{person}{Guandao Yang},
  \bibinfo{person}{Qingqing Zhao}, \bibinfo{person}{Xi Deng},
  \bibinfo{person}{Leonidas Guibas}, \bibinfo{person}{Bharath Hariharan}, {and}
  \bibinfo{person}{Gordon Wetzstein}.} \bibinfo{year}{2024}\natexlab{}.
\newblock \showarticletitle{Neural Control Variates with Automatic
  Integration}. In \bibinfo{booktitle}{\emph{ACM SIGGRAPH 2024 Conference
  Papers}}. \bibinfo{pages}{1--9}.
\newblock


\bibitem[Madan and Levin(2022)]%
        {madan2022smoothdists}
\bibfield{author}{\bibinfo{person}{Abhishek Madan} {and}
  \bibinfo{person}{David~I.W. Levin}.} \bibinfo{year}{2022}\natexlab{}.
\newblock \showarticletitle{Fast Evaluation of Smooth Distance Constraints on
  Co-Dimensional Geometry}.
\newblock \bibinfo{journal}{\emph{ACM Trans. Graph.}} \bibinfo{volume}{41},
  \bibinfo{number}{4} (\bibinfo{year}{2022}).
\newblock


\bibitem[Makino(1990)]%
        {makino1990vectorization}
\bibfield{author}{\bibinfo{person}{Junichiro Makino}.}
  \bibinfo{year}{1990}\natexlab{}.
\newblock \showarticletitle{Vectorization of a treecode}.
\newblock \bibinfo{journal}{\emph{J. Comput. Phys.}} \bibinfo{volume}{87},
  \bibinfo{number}{1} (\bibinfo{year}{1990}), \bibinfo{pages}{148--160}.
\newblock


\bibitem[Martinsson(2019)]%
        {martinsson2019fast}
\bibfield{author}{\bibinfo{person}{Per-Gunnar Martinsson}.}
  \bibinfo{year}{2019}\natexlab{}.
\newblock \bibinfo{booktitle}{\emph{Fast direct solvers for elliptic PDEs}}.
\newblock \bibinfo{publisher}{SIAM}.
\newblock


\bibitem[McLeish(2011)]%
        {Mcleish2011}
\bibfield{author}{\bibinfo{person}{Don McLeish}.}
  \bibinfo{year}{2011}\natexlab{}.
\newblock \showarticletitle{A general method for debiasing a Monte Carlo
  estimator}.
\newblock \bibinfo{journal}{\emph{Monte Carlo Methods and Applications}}
  \bibinfo{volume}{17}, \bibinfo{number}{4} (\bibinfo{year}{2011}),
  \bibinfo{pages}{301--315}.
\newblock
\href{https://doi.org/doi:10.1515/mcma.2011.013}{doi:\nolinkurl{doi:10.1515/mcma.2011.013}}


\bibitem[Metropolis and Ulam(1949)]%
        {metropolis1949monte}
\bibfield{author}{\bibinfo{person}{Nicholas Metropolis} {and}
  \bibinfo{person}{Stanislaw Ulam}.} \bibinfo{year}{1949}\natexlab{}.
\newblock \showarticletitle{The monte carlo method}.
\newblock \bibinfo{journal}{\emph{Journal of the American statistical
  association}} \bibinfo{volume}{44}, \bibinfo{number}{247}
  (\bibinfo{year}{1949}), \bibinfo{pages}{335--341}.
\newblock


\bibitem[Misso et~al\mbox{.}(2022)]%
        {misso2022unbiased}
\bibfield{author}{\bibinfo{person}{Zackary Misso}, \bibinfo{person}{Benedikt
  Bitterli}, \bibinfo{person}{Iliyan Georgiev}, {and} \bibinfo{person}{Wojciech
  Jarosz}.} \bibinfo{year}{2022}\natexlab{}.
\newblock \showarticletitle{Unbiased and consistent rendering using biased
  estimators}.
\newblock \bibinfo{journal}{\emph{ACM Transactions on Graphics (TOG)}}
  \bibinfo{volume}{41}, \bibinfo{number}{4} (\bibinfo{year}{2022}),
  \bibinfo{pages}{1--13}.
\newblock


\bibitem[M{\"u}ller et~al\mbox{.}(2020)]%
        {muller2020neural}
\bibfield{author}{\bibinfo{person}{Thomas M{\"u}ller}, \bibinfo{person}{Fabrice
  Rousselle}, \bibinfo{person}{Alexander Keller}, {and} \bibinfo{person}{Jan
  Nov{\'a}k}.} \bibinfo{year}{2020}\natexlab{}.
\newblock \showarticletitle{Neural control variates}.
\newblock \bibinfo{journal}{\emph{ACM Transactions on Graphics (TOG)}}
  \bibinfo{volume}{39}, \bibinfo{number}{6} (\bibinfo{year}{2020}),
  \bibinfo{pages}{1--19}.
\newblock


\bibitem[Museth(2021)]%
        {museth21nanovdb}
\bibfield{author}{\bibinfo{person}{Ken Museth}.}
  \bibinfo{year}{2021}\natexlab{}.
\newblock \showarticletitle{NanoVDB: A GPU-Friendly and Portable VDB Data
  Structure For Real-Time Rendering And Simulation}. In
  \bibinfo{booktitle}{\emph{ACM SIGGRAPH 2021 Talks}} (Virtual Event, USA)
  \emph{(\bibinfo{series}{SIGGRAPH '21})}. \bibinfo{publisher}{Association for
  Computing Machinery}, \bibinfo{address}{New York, NY, USA}, Article
  \bibinfo{articleno}{1}, \bibinfo{numpages}{2}~pages.
\newblock
\showISBNx{9781450383738}
\href{https://doi.org/10.1145/3450623.3464653}{doi:\nolinkurl{10.1145/3450623.3464653}}


\bibitem[Myles et~al\mbox{.}(2014)]%
        {MPZ14}
\bibfield{author}{\bibinfo{person}{Ashish Myles}, \bibinfo{person}{Nico
  Pietroni}, {and} \bibinfo{person}{Denis Zorin}.}
  \bibinfo{year}{2014}\natexlab{}.
\newblock \showarticletitle{Robust Field-aligned Global Parametrization}.
\newblock \bibinfo{journal}{\emph{ACM Trans. on Graphics - Siggraph 2014}}
  \bibinfo{volume}{33}, \bibinfo{number}{4} (\bibinfo{year}{2014}),
  \bibinfo{pages}{Article No. 135}.
\newblock
\urldef\tempurl%
\url{http://vcg.isti.cnr.it/Publications/2014/MPZ14}
\showURL{%
\tempurl}


\bibitem[Nyland et~al\mbox{.}(2007)]%
        {nyland2007fast}
\bibfield{author}{\bibinfo{person}{Lars Nyland}, \bibinfo{person}{Mark Harris},
  {and} \bibinfo{person}{Jan Prins}.} \bibinfo{year}{2007}\natexlab{}.
\newblock \showarticletitle{Fast n-body simulation with cuda}.
\newblock \bibinfo{journal}{\emph{GPU gems}}  \bibinfo{volume}{3}
  (\bibinfo{year}{2007}), \bibinfo{pages}{62--66}.
\newblock


\bibitem[Owen(2013)]%
        {mcbook}
\bibfield{author}{\bibinfo{person}{Art~B. Owen}.}
  \bibinfo{year}{2013}\natexlab{}.
\newblock \bibinfo{booktitle}{\emph{Monte Carlo theory, methods and examples}}.
\newblock \bibinfo{publisher}{\url{https://artowen.su.domains/mc/}}.
\newblock


\bibitem[Parker et~al\mbox{.}(2010)]%
        {parker2010optix}
\bibfield{author}{\bibinfo{person}{Steven~G Parker}, \bibinfo{person}{James
  Bigler}, \bibinfo{person}{Andreas Dietrich}, \bibinfo{person}{Heiko
  Friedrich}, \bibinfo{person}{Jared Hoberock}, \bibinfo{person}{David Luebke},
  \bibinfo{person}{David McAllister}, \bibinfo{person}{Morgan McGuire},
  \bibinfo{person}{Keith Morley}, \bibinfo{person}{Austin Robison},
  {et~al\mbox{.}}} \bibinfo{year}{2010}\natexlab{}.
\newblock \showarticletitle{Optix: a general purpose ray tracing engine}.
\newblock \bibinfo{journal}{\emph{Acm transactions on graphics (tog)}}
  \bibinfo{volume}{29}, \bibinfo{number}{4} (\bibinfo{year}{2010}),
  \bibinfo{pages}{1--13}.
\newblock


\bibitem[Pharr et~al\mbox{.}(2023)]%
        {pharr2023physically}
\bibfield{author}{\bibinfo{person}{Matt Pharr}, \bibinfo{person}{Wenzel Jakob},
  {and} \bibinfo{person}{Greg Humphreys}.} \bibinfo{year}{2023}\natexlab{}.
\newblock \bibinfo{booktitle}{\emph{Physically based rendering: From theory to
  implementation}}.
\newblock \bibinfo{publisher}{MIT Press}.
\newblock


\bibitem[Qu and James(2021)]%
        {qu2021fast}
\bibfield{author}{\bibinfo{person}{Ante Qu} {and} \bibinfo{person}{Doug~L
  James}.} \bibinfo{year}{2021}\natexlab{}.
\newblock \showarticletitle{Fast linking numbers for topology verification of
  loopy structures.}
\newblock \bibinfo{journal}{\emph{ACM Trans. Graph.}} \bibinfo{volume}{40},
  \bibinfo{number}{4} (\bibinfo{year}{2021}), \bibinfo{pages}{106--1}.
\newblock


\bibitem[Rath et~al\mbox{.}(2022)]%
        {Rath2022}
\bibfield{author}{\bibinfo{person}{Alexander Rath}, \bibinfo{person}{Pascal
  Grittmann}, \bibinfo{person}{Sebastian Herholz}, \bibinfo{person}{Philippe
  Weier}, {and} \bibinfo{person}{Philipp Slusallek}.}
  \bibinfo{year}{2022}\natexlab{}.
\newblock \showarticletitle{EARS: Efficiency-Aware Russian Roulette and
  Splitting}.
\newblock \bibinfo{journal}{\emph{ACM Transactions on Graphics (Proceedings of
  SIGGRAPH 2022)}} \bibinfo{volume}{41}, \bibinfo{number}{4}, Article
  \bibinfo{articleno}{81} (\bibinfo{date}{jul} \bibinfo{year}{2022}),
  \bibinfo{numpages}{14}~pages.
\newblock
\href{https://doi.org/10.1145/3528223.3530168}{doi:\nolinkurl{10.1145/3528223.3530168}}


\bibitem[Rhee and Glynn(2015)]%
        {rhee2015unbiased}
\bibfield{author}{\bibinfo{person}{Chang-han Rhee} {and}
  \bibinfo{person}{Peter~W Glynn}.} \bibinfo{year}{2015}\natexlab{}.
\newblock \showarticletitle{Unbiased estimation with square root convergence
  for SDE models}.
\newblock \bibinfo{journal}{\emph{Operations Research}} \bibinfo{volume}{63},
  \bibinfo{number}{5} (\bibinfo{year}{2015}), \bibinfo{pages}{1026--1043}.
\newblock


\bibitem[Rioux-Lavoie et~al\mbox{.}(2022)]%
        {rioux2022monte}
\bibfield{author}{\bibinfo{person}{Damien Rioux-Lavoie},
  \bibinfo{person}{Ryusuke Sugimoto}, \bibinfo{person}{T{\"u}may {\"O}zdemir},
  \bibinfo{person}{Naoharu~H Shimada}, \bibinfo{person}{Christopher Batty},
  \bibinfo{person}{Derek Nowrouzezahrai}, {and} \bibinfo{person}{Toshiya
  Hachisuka}.} \bibinfo{year}{2022}\natexlab{}.
\newblock \showarticletitle{A monte carlo method for fluid simulation}.
\newblock \bibinfo{journal}{\emph{ACM Transactions on Graphics (TOG)}}
  \bibinfo{volume}{41}, \bibinfo{number}{6} (\bibinfo{year}{2022}),
  \bibinfo{pages}{1--16}.
\newblock


\bibitem[Sawhney and Crane(2020)]%
        {sawhney2020monte}
\bibfield{author}{\bibinfo{person}{Rohan Sawhney} {and} \bibinfo{person}{Keenan
  Crane}.} \bibinfo{year}{2020}\natexlab{}.
\newblock \showarticletitle{Monte Carlo geometry processing: A grid-free
  approach to PDE-based methods on volumetric domains}.
\newblock \bibinfo{journal}{\emph{ACM Transactions on Graphics}}
  \bibinfo{volume}{39}, \bibinfo{number}{4} (\bibinfo{year}{2020}).
\newblock


\bibitem[Shirley et~al\mbox{.}(1996)]%
        {shirley1996monte}
\bibfield{author}{\bibinfo{person}{Peter Shirley}, \bibinfo{person}{Changyaw
  Wang}, {and} \bibinfo{person}{Kurt Zimmerman}.}
  \bibinfo{year}{1996}\natexlab{}.
\newblock \showarticletitle{Monte carlo techniques for direct lighting
  calculations}.
\newblock \bibinfo{journal}{\emph{ACM Transactions on Graphics (TOG)}}
  \bibinfo{volume}{15}, \bibinfo{number}{1} (\bibinfo{year}{1996}),
  \bibinfo{pages}{1--36}.
\newblock


\bibitem[Sugimoto et~al\mbox{.}(2023)]%
        {sugimoto2023practical}
\bibfield{author}{\bibinfo{person}{Ryusuke Sugimoto}, \bibinfo{person}{Terry
  Chen}, \bibinfo{person}{Yiti Jiang}, \bibinfo{person}{Christopher Batty},
  {and} \bibinfo{person}{Toshiya Hachisuka}.} \bibinfo{year}{2023}\natexlab{}.
\newblock \showarticletitle{A practical walk-on-boundary method for boundary
  value problems}.
\newblock \bibinfo{journal}{\emph{ACM Transactions on Graphics (TOG)}}
  \bibinfo{volume}{42}, \bibinfo{number}{4} (\bibinfo{year}{2023}),
  \bibinfo{pages}{1--16}.
\newblock


\bibitem[Szirmay{-}Kalos et~al\mbox{.}(2005)]%
        {Szirmay-KalosAS05}
\bibfield{author}{\bibinfo{person}{L{\'{a}}szl{\'{o}} Szirmay{-}Kalos},
  \bibinfo{person}{Gy{\"{o}}rgy Antal}, {and} \bibinfo{person}{Mateu Sbert}.}
  \bibinfo{year}{2005}\natexlab{}.
\newblock \showarticletitle{Go with the Winners Strategy in Path Tracing}.
\newblock \bibinfo{journal}{\emph{J. {WSCG}}} \bibinfo{volume}{13},
  \bibinfo{number}{2} (\bibinfo{year}{2005}), \bibinfo{pages}{49--56}.
\newblock
\urldef\tempurl%
\url{http://wscg.zcu.cz/wscg2005/Papers\_2005/Journal/!WSCG2005\_Journal\_Final.pdf}
\showURL{%
\tempurl}


\bibitem[Vaidyanathan et~al\mbox{.}(2019)]%
        {vaidyanathan2019wide}
\bibfield{author}{\bibinfo{person}{Karthik Vaidyanathan}, \bibinfo{person}{Sven
  Woop}, {and} \bibinfo{person}{Carsten Benthin}.}
  \bibinfo{year}{2019}\natexlab{}.
\newblock \showarticletitle{Wide BVH traversal with a short stack}. In
  \bibinfo{booktitle}{\emph{Proceedings of the Conference on High-Performance
  Graphics}}. \bibinfo{pages}{15--19}.
\newblock


\bibitem[Vicini et~al\mbox{.}(2019)]%
        {vicini2019denoising}
\bibfield{author}{\bibinfo{person}{Delio Vicini}, \bibinfo{person}{David
  Adler}, \bibinfo{person}{Jan Nov{\'a}k}, \bibinfo{person}{Fabrice Rousselle},
  {and} \bibinfo{person}{Brent Burley}.} \bibinfo{year}{2019}\natexlab{}.
\newblock \showarticletitle{Denoising deep monte carlo renderings}. In
  \bibinfo{booktitle}{\emph{Computer Graphics Forum}},
  Vol.~\bibinfo{volume}{38}. Wiley Online Library, \bibinfo{pages}{316--327}.
\newblock


\bibitem[Vorba and K{\v{r}}iv{\'{a}}nek(2016)]%
        {Vorba:2016:ARR}
\bibfield{author}{\bibinfo{person}{Ji{\v{r}}{\'{i}} Vorba} {and}
  \bibinfo{person}{Jaroslav K{\v{r}}iv{\'{a}}nek}.}
  \bibinfo{year}{2016}\natexlab{}.
\newblock \showarticletitle{Adjoint-driven Russian Roulette and Splitting in
  Light Transport Simulation}.
\newblock \bibinfo{journal}{\emph{ACM Trans. Graph.}} \bibinfo{volume}{35},
  \bibinfo{number}{4}, Article \bibinfo{articleno}{42} (\bibinfo{date}{July}
  \bibinfo{year}{2016}), \bibinfo{numpages}{11}~pages.
\newblock
\showISSN{0730-0301}
\href{https://doi.org/10.1145/2897824.2925912}{doi:\nolinkurl{10.1145/2897824.2925912}}


\bibitem[Wald et~al\mbox{.}(2001)]%
        {wald2001interactive}
\bibfield{author}{\bibinfo{person}{Ingo Wald}, \bibinfo{person}{Philipp
  Slusallek}, \bibinfo{person}{Carsten Benthin}, {and} \bibinfo{person}{Markus
  Wagner}.} \bibinfo{year}{2001}\natexlab{}.
\newblock \showarticletitle{Interactive rendering with coherent ray tracing}.
  In \bibinfo{booktitle}{\emph{Computer graphics forum}},
  Vol.~\bibinfo{volume}{20}. Wiley Online Library, \bibinfo{pages}{153--165}.
\newblock


\bibitem[Walter et~al\mbox{.}(2005)]%
        {walter2005lightcuts}
\bibfield{author}{\bibinfo{person}{Bruce Walter}, \bibinfo{person}{Sebastian
  Fernandez}, \bibinfo{person}{Adam Arbree}, \bibinfo{person}{Kavita Bala},
  \bibinfo{person}{Michael Donikian}, {and} \bibinfo{person}{Donald~P
  Greenberg}.} \bibinfo{year}{2005}\natexlab{}.
\newblock \showarticletitle{Lightcuts: a scalable approach to illumination}.
\newblock In \bibinfo{booktitle}{\emph{ACM SIGGRAPH 2005 Papers}}.
  \bibinfo{pages}{1098--1107}.
\newblock


\bibitem[Wang et~al\mbox{.}(2021)]%
        {wang2021learning}
\bibfield{author}{\bibinfo{person}{Yu-Chen Wang}, \bibinfo{person}{Yu-Ting Wu},
  \bibinfo{person}{Tzu-Mao Li}, {and} \bibinfo{person}{Yung-Yu Chuang}.}
  \bibinfo{year}{2021}\natexlab{}.
\newblock \showarticletitle{Learning to cluster for rendering with many
  lights}.
\newblock \bibinfo{journal}{\emph{ACM Transactions on Graphics (TOG)}}
  \bibinfo{volume}{40}, \bibinfo{number}{6} (\bibinfo{year}{2021}),
  \bibinfo{pages}{1--10}.
\newblock


\bibitem[Williams et~al\mbox{.}(2024)]%
        {williams2024fVDB}
\bibfield{author}{\bibinfo{person}{Francis Williams}, \bibinfo{person}{Jiahui
  Huang}, \bibinfo{person}{Jonathan Swartz}, \bibinfo{person}{Gergely Klar},
  \bibinfo{person}{Vijay Thakkar}, \bibinfo{person}{Matthew Cong},
  \bibinfo{person}{Xuanchi Ren}, \bibinfo{person}{Ruilong Li},
  \bibinfo{person}{Clement Fuji-Tsang}, \bibinfo{person}{Sanja Fidler},
  \bibinfo{person}{Eftychios Sifakis}, {and} \bibinfo{person}{Ken Museth}.}
  \bibinfo{year}{2024}\natexlab{}.
\newblock \showarticletitle{fVDB: A Deep-Learning Framework for Sparse,
  Large-Scale, and High-Performance Spatial Intelligence}.
\newblock \bibinfo{journal}{\emph{ACM Transactions on Graphics (SIGGRAPH)}}
  \bibinfo{volume}{43}, \bibinfo{number}{4}, Article \bibinfo{articleno}{133}
  (\bibinfo{date}{July} \bibinfo{year}{2024}), \bibinfo{numpages}{15}~pages.
\newblock
\href{https://doi.org/10.1145/3658226}{doi:\nolinkurl{10.1145/3658226}}


\bibitem[Wyman et~al\mbox{.}(2023)]%
        {Wyman:2023:Gentle}
\bibfield{author}{\bibinfo{person}{Chris Wyman}, \bibinfo{person}{Markus
  Kettunen}, \bibinfo{person}{Daqi Lin}, \bibinfo{person}{Benedikt Bitterli},
  \bibinfo{person}{Cem Yuksel}, \bibinfo{person}{Wojciech Jarosz},
  \bibinfo{person}{Pawel Kozlowski}, {and} \bibinfo{person}{Giovanni~De
  Francesco}.} \bibinfo{year}{2023}\natexlab{}.
\newblock \showarticletitle{A Gentle Introduction to ReSTIR: Path Reuse in
  Real-time}. In \bibinfo{booktitle}{\emph{ACM SIGGRAPH 2023 Courses}} (Los
  Angeles, California).
\newblock
\href{https://doi.org/10.1145/3587423.3595511}{doi:\nolinkurl{10.1145/3587423.3595511}}


\bibitem[Ylitie et~al\mbox{.}(2017)]%
        {ylitie2017efficient}
\bibfield{author}{\bibinfo{person}{Henri Ylitie}, \bibinfo{person}{Tero
  Karras}, {and} \bibinfo{person}{Samuli Laine}.}
  \bibinfo{year}{2017}\natexlab{}.
\newblock \showarticletitle{Efficient incoherent ray traversal on GPUs through
  compressed wide BVHs}.
\newblock In \bibinfo{booktitle}{\emph{Proceedings of High Performance
  Graphics}}. \bibinfo{publisher}{ACM}, \bibinfo{pages}{1--13}.
\newblock


\bibitem[Yu et~al\mbox{.}(2021a)]%
        {yu2021rs}
\bibfield{author}{\bibinfo{person}{Chris Yu}, \bibinfo{person}{Caleb
  Brakensiek}, \bibinfo{person}{Henrik Schumacher}, {and}
  \bibinfo{person}{Keenan Crane}.} \bibinfo{year}{2021}\natexlab{a}.
\newblock \showarticletitle{Repulsive Surfaces}.
\newblock \bibinfo{journal}{\emph{ACM Trans. Graph.}} \bibinfo{volume}{40},
  \bibinfo{number}{6} (\bibinfo{year}{2021}).
\newblock


\bibitem[Yu et~al\mbox{.}(2021b)]%
        {yu2021rc}
\bibfield{author}{\bibinfo{person}{Chris Yu}, \bibinfo{person}{Henrik
  Schumacher}, {and} \bibinfo{person}{Keenan Crane}.}
  \bibinfo{year}{2021}\natexlab{b}.
\newblock \showarticletitle{Repulsive curves}.
\newblock \bibinfo{journal}{\emph{ACM Transactions on Graphics (TOG)}}
  \bibinfo{volume}{40}, \bibinfo{number}{2} (\bibinfo{year}{2021}),
  \bibinfo{pages}{1--21}.
\newblock


\bibitem[Yuksel(2019)]%
        {yuksel2019stochastic}
\bibfield{author}{\bibinfo{person}{Cem Yuksel}.}
  \bibinfo{year}{2019}\natexlab{}.
\newblock \showarticletitle{Stochastic lightcuts}. In
  \bibinfo{booktitle}{\emph{Proceedings of the Conference on High-Performance
  Graphics}}. \bibinfo{pages}{27--32}.
\newblock


\end{thebibliography}

\appendix
\section{Proof of Unbiasedness}
Here we prove Theorem~\ref{thm:unbiased}.
A restatement of the theorem is as follows:
\begin{theorem*}
Eq.~\ref{eq:estimator} is an unbiased estimator of Eq.~\ref{eq:kernel_sum}; that is, $E \left[ \hat{F}_1(\bq) \right] = F(\bq)$.
\end{theorem*}

\begin{proof}
Expanding $E \left[ \hat{F}_1(\bq) \right]$:
\begin{align*}
    E \left[ \hat{F}_1(\bq) \right] &= \sum_{i = 1}^M p(I = i)\sum_{\ell=0}^{d_i} p(K = \ell) \Biggl( \mm_{i,0}f(\com_{i,0}, \bq) \quad + \\
                                    &\quad \sum_{k=1}^\ell \frac{\Delta_{i,k-1}}{p(I \in \Tree_{i,k-1})p(K \ge k)} \Biggr) \\
                                    &= \sum_{i = 1}^M \sum_{\ell=0}^{d_i} p(I = i)p(K = \ell)\Biggl( \sum_{j \in \Tree_{i,0}} m_j f(\com_{I,0}, \bq) \quad + \\
                                    &\quad \sum_{k=1}^\ell \sum_{c \in \Children(\Tree_{i,k-1})} \sum_{j \in \Tree_c} \frac{ m_j( f(\com_c, \bq) - f(\com_{i,k-1}, \bq))}{p(I \in \Tree_{i,k-1})p(K \ge k)} \Biggr) \\
                                    &= \sum_{i=1}^M m_i f(\com_0, \bq) + \sum_{i = 1}^M \sum_{\ell=0}^{d_i} \sum_{k=1}^\ell \sum_{c \in \Children(\Tree_{i,k-1})} \sum_{j \in \Tree_c} \\
                                    & \qquad \frac{ p(I = i)p(K = \ell) m_j( f(\com_c, \bq) - f(\com_{i,k-1}, \bq))}{p(I \in \Tree_{i,k-1})p(K \ge k)} \\
                                    &= \sum_{i = 1}^M \Biggl( m_i f(\com_0, \bq) \quad + \\
                                    &\quad \sum_{k = 1}^{d_i} \cancelto{1}{\frac{ \left( \sum_{j \in \Tree_{i,k-1}}p(I = j) \right) \left( \sum_{\ell = k}^{d_i} p(K = \ell) \right) }{p(I \in \Tree_{i,k-1})p(K \ge k)}} \\
                                    & \qquad m_i( f(\com_{i,k}, \bq) - f(\com_{i,k-1}, \bq)) \Biggr) \\
                                    &= \sum_{i = 1}^M F_{\Path_i}(\bq) \\ 
                                    &= F(\bq)
\end{align*}
The third equality comes from the fact that every path has a common root (i.e., $T_{i,0} = T_0$ for all $i$) and $\sum_{i = 1}^M \sum_{\ell=0}^{d_i} p(I=i)p(K=\ell) = 1$, and the fourth equality comes from rearranging and grouping like terms from contribution swaps along different paths with common prefixes.
\end{proof}

\section{Pseudocode}
Pseudocode for the Barnes-Hut approximation, and our path sample estimator, are given below.

\begin{algorithm}
    \caption{Barnes-Hut approximation, \texttt{barnesHut}}\label{alg:bh}
    \SetAlgoLined
    \LinesNumbered
    \SetKwInOut{Input}{Inputs}
    \SetKwInOut{Output}{Outputs}
    \Input{Tree node $\Tree_a$, kernel function $f$, query point $\bq$, approximation parameter $\beta$}
    \Output{Approximation of Eq.~\ref{eq:kernel_sum} (i.e., Eq.~\ref{eq:bh_sum})}
    \BlankLine
    \CommentSty{$\Tree_a$ has total mass $\mm_a$ and center of mass $\com_a$}\;
    $F \leftarrow 0$\;
    $\FF \leftarrow \frac{\| \bq - \com_a \|}{| \bbox_a |}$\;
    \eIf{$\Tree_a$ is a leaf or $\FF \ge \beta$}{
        $F \leftarrow \mm_a f(\com_a, \bq)$\;
    }{
        \For{$\Tree_b \in \Children(\Tree_a)$}{
            $F \leftarrow F + \texttt{barnesHut}(\Tree_b, f, \bq, \beta)$\;
        }
    }
    Return $F$\;
\end{algorithm}

\begin{algorithm}
    \caption{Stratified path sample estimator, \texttt{pathSampleEstimator}}\label{alg:path}
    \SetAlgoLined
    \LinesNumbered
    \SetKwInOut{Input}{Inputs}
    \SetKwInOut{Output}{Outputs}
    \Input{Tree root $\Tree_0$, kernel function $f$, query point $\bq$, samples per subdomain $S$}
    \Output{Approximation of Eq.~\ref{eq:kernel_sum}}
    \BlankLine
    $F \leftarrow 0$\;
    \For{$\Tree_a \in \Children(\Tree_0)$}{
        \CommentSty{$\Tree_a$ has total mass $\mm_a$ and center of mass $\com_a$}\;
        $CV_a \leftarrow \mm_a f(\com_a, \bq)$\;
        $F_a \leftarrow 0$\;
        \For{$s \in 1 \ldots S$}{
            $I \leftarrow \texttt{uniformSamplePath}(\Tree_a)$\;
            $k \leftarrow 0$\;
            $\prr \leftarrow 1$\;
            $\Tree_{I,k} \leftarrow \Tree_a$\;
            \While{$\Tree_{I,k}$ is not a leaf}{
                $p_{I,k} \leftarrow$ Eq.~\ref{eq:rr_prob}\;
                \If{$\texttt{rand}() \ge p_{I,k}$}{
                    break;
                }
                $\Delta_{I,k} \leftarrow$ Eq.~\ref{eq:contrib_swap}\;
                $\pa \leftarrow \frac{| \Tree_{I,k} |}{| \Tree_a |}$ ; $\prr \leftarrow \prr \cdot p_{I,k}$\;
                $F_a \leftarrow F_a + \frac{\Delta_{I,k}}{\pa \cdot \prr}$\;
                $\Tree_{I,k} \leftarrow$ child of $\Tree_{I,k}$ containing $I$ ; $k \leftarrow k + 1$\;
            }
        }
        $F \leftarrow F + CV_a + \frac{1}{S}F_a$\;
    }
    return $F$\;
\end{algorithm}

\begin{figure*}
\includegraphics[width=\textwidth]{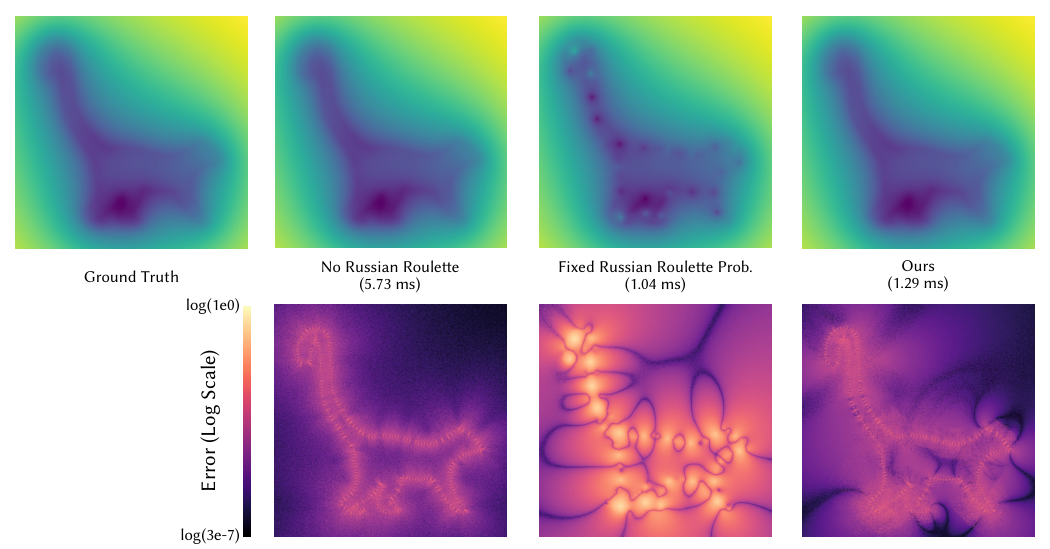}
\caption{
Compared to the ground truth (top, left), an estimator with Russian roulette disabled (top, middle left) and our estimator (top, right) produce virtually identical Coulomb potential fields from a source set sampled from a lemur polyline, while using a fixed Russian roulette probability (top, middle right) has significant errors. The log error of each estimator is shown in the bottom row. Despite our method having slightly higher log error than a disabled Russian roulette estimator, it is roughly 5 times faster to evaluate on a GPU, only being slightly slower than the simple fixed-probability Russian roulette scheme.
}
\label{fig:rr_supp}
\end{figure*}
\section{Russian Roulette Comparison}
Russian roulette schemes are a technique used to improve an estimator's efficiency (the reciprocal of variance times computational cost), with the tradeoff of increasing variance.
In other words, good Russian roulette schemes should indirectly lead to variance reduction by allowing more samples to be taken for the same computational budget, or equivalently, a slight increase in variance in less time for the same sample budget.
We investigated the quality of our Russian roulette algorithm described in the main paper by comparing it to alternative estimators without Russian roulette (i.e., all paths are traversed to the leaves), and using a fixed probability of $\frac{1}{2}$.
The results are shown in Fig.~\ref{fig:rr_supp}.
Overall, our method achieves results comparable to an estimator without Russian roulette while being 5 times faster.

\end{document}